\documentclass[12pt]{article}

\usepackage[margin=1in]{geometry}  % set the margins to 1in on all sides
\usepackage{graphicx}              % to include figures
\usepackage{amsmath}               % great math stuff
\usepackage{amsfonts}              % for blackboard bold, etc
\usepackage{amsthm}                % better theorem environments
\usepackage{epsfig}
\usepackage{float}
\newtheorem{thm}{Theorem}[section]

\theoremstyle{definition}
\newtheorem{example}{Example}[section]

\begin{document}

\title{CRT Based Spectral Convolution in Binary Fields}

\author{Muhammad Asad Khan, Sajid Saleem, Amir A Khan\\ 
School of Electrical Engineering and Computer Science \\
National University of Sciences and Technology\\
Islamabad, Pakistan \\
%\texttt{* asad.khan@seecs.edu.pk}
{\textit {\{asad.khan, sajid.saleem, amir.ali\}@seecs.edu.pk}}
}
\maketitle

\begin{abstract}
  In this paper, new results on convolution of spectral components in binary fields have been presented for combiatorial sequences. A novel method of convolution of DFT points through Chinese Remainder  Theorem (CRT) is presented which has lower complexity as compared to  known methods of spectral point computations. Exploring the inherent structures in cyclic nature of finite fields, certain fixed mappings between the spectral components from composite fields to their decomposed subfield components has been illustrated which are significant for analysis of combiner generators. Complexity estimations of our CRT based methodology of convolutions in binary fields proves that our proposed method is far efficient as  comparised to to existing methods of DFT computations for convolving sequences in frequency domain. 
\end{abstract}

\section{Introduction}

Let $GF(2^n)$ be a Galois field and the integer $N$ divides $2^n - 1$.
From a known theory of digital signal processing, the convolution in time domain corresponds to multiplication in frequency domain as \cite{reed1975use}:

\begin{equation}
u_p = \sum^{d-1}_{i=0} a_{i}b_{(p-i)}
\end{equation}   
 is equivalent to spectral point multiplication in binary fields as:
 
 \begin{equation}\label{DFT-product}
U_k = A_kB_k
\end{equation}
 
 where \textbf{U}$_k$, \textbf{A}$_k$ and \textbf{B}$_k$ are Fourier components of \textbf{u}$_i$, \textbf{a}$_i$  and \textbf{b}$_i$, for $0\leq i \leq N-1$ and $0\leq k \leq N-1$.
 
Similarly treating the problem conversely, multiplication in time domain is equivalent to convolution in frequency domain. As bit wise multiplication is a fundamental block for any cipher design, frequency domain analysis of these cryptographic algorithms involves convolution theory invariably. 

Consider a \textit{N}-tupple sequence \textbf{u}$ ^N$ = [ $u_0,u_1,\cdots ,u_{N-1}$]  which is bit-wise product $u_i = a_i .b_i $($i=0,1,\cdots,N-1$) of two \textit{N}-tupple sequences \textbf{a}$ ^N$ = [$a_0,a_1,\cdots ,a_{N-1}$] and \textbf{b}$ ^N$ = [$b_0,b_1,\cdots ,b_{N-1}$]. From \cite{massey1994fourier},   
the frequency domain \textit{N}-tupple \textbf{U}$ ^N$ = [ $U_0,U_1,\cdots ,U_{N-1}$] is cyclic convolution of of $A^N$ and $B^N$ as:
\begin{equation}\label{conv-basic}
u_j = \frac{1}{N^*}\sum_{k=0}^{N-1}A_{(j-k)}B_k \;\;\;\;\;for\;\;j=0,1,\cdots,N-1.
\end{equation}
 
The frequency domain scenerio of two \textit{N}-tupple sequences belonging to same binary fields is simple to relate with Equation~(\ref{DFT-product}) and~(\ref{conv-basic}). However, when  sequences belong to different binary fields, the relationship becomes little complex.  Consider an LFSR sequence \textbf{a}$_t \;\in GF(2^p)$ having a period $n_1|(2^p-1)$ and another LFSR sequence \textbf{b}$_t \;\in GF(2^q)$ with period $n_2|(2^q-1)$. The associated DFT components of  these sequences are represented in terms of powers of primitive elements of their respective binary fields as $\alpha \in GF(2^p)$  and $\beta \in GF(2^q)$ for instance. Computing product of terms $A_{(j-k)}B_k$ directly in Equation~(\ref{conv-basic}), represented in terms of primitive elements $\alpha \in GF(2^p)$  and $\beta \in GF(2^q)$ and belonging to different binary fields, is not simple and much of the details have not been discussed even in~\cite{ding1996chinese},\cite{Blahut1983errorcontrol}, \cite{golic1995linear} and \cite{golomb2005signal}. In this paper, we have dicussed this apsect explicitly and presented  new method of computing spectral convolutions in binary fields.    
 
Chinese Remainder Theorem (CRT) based computations of convolution between elements belonging to different binary fields is introduced as our novel finding in this paper. In Section-2, we have presented main idea of our work followed by illustration through an example in binary fields. Section-3 covers discussion on application of CRT based DFT computations in analysis of combinatorial sequences through subspace decomposition. New results are demonstrated through small examples for clarity of context. The computational efficiency of our CRT based proposed method in comparison to existing method of DFT computations in binary fields has been discussed in Section-4.  The paper is finally concluded in section-5.
 
 \section{Spectral Convolution in Binary Fields and CRT}
 
Over the past few years,  spectral analysis of LFSR based sequence generators is introduced as a promising idea in cryptanalysis of stream ciphers and fundamental in the series is discrete fast fourier spectra  attacks on filter generators~\cite{gong2011fast}. In case of combiner generators, typical designs involve number of LFSRs based on  primitive connection polynomials having periods co-prime to each other for attaining maximum key-stream periods.  In this case, when number of involved binary fields increase, convolution of spectral components represented in elements belonging to different binary fields is inevitable. To illustrate this, let we consider two sequences  \textbf{a}$_i \;\in\; GF(2^p)$ and \textbf{b}$_i\;\in\; GF(2^q)$ based on primitive elements and $2^{(p-1)}$ and $2^{(q-1)}$ are coprime to each other where $p < q$. If we consider a simplest case of bit wise product, being part of any non-linear boolean mapping, as $s_i = a_i .b_i $ where ($i=0,1,\cdots,N-1$) such that $(p-1)|N$ and $(q-1)|N$, their fourier transform is determined using the relations:

\begin{equation}\label{DFT-A eq}
  A_{k} = \sum_{t=0}^{p-1} a_{t}\alpha^{-tk} , \;\;\;  k = 0,1,2,.....,p-1
\end{equation}
%\end{definition}
and
\begin{equation}\label{DFT-B eq}
  B_{k} = \sum_{t=0}^{q-1} b_{t}\beta^{-tk} , \;\;\;  k = 0,1,2,.....,q-1
\end{equation}

 where $A_k$ and $B_k$ are $k$-th frequency components of DFTs of \textbf{a}$_t \;\in\; GF(2^p)$ and \textbf{b}$_t \;\in\; GF(2^q)$ where $\alpha$ and $\beta$ are the primitive elements within their respective fields; generators of $GF(2^p)$ and $GF(2^q)$ with periods $(p-1)$ and $(q-1)$ respectively~\cite{pollard1971fast}. 
 For \textbf{u}$_t$, Berlekamp-Massey algorithm~\cite{golomb2005signal} gives associated minimum polynomail $g(x) \in GF(2^n)$ of \textbf{u}$_t$. Classically, DFT of \textbf{u}$_t$  is then taken with respect to $g(x) \in GF(2^n)$ as:
 \begin{equation}\label{DFT-S eq}
  U_{k} = \sum_{t=0}^{q-1} u_{t}\lambda^{-tk} , \;\;\;  k = 0,1,2,.....,N-1
\end{equation}
For this case, we need $N$ bits of stream \textbf{u}$_t$ for computing each component of \textbf{U} and all computations are in $GF(2^n)$. For practical scenerios of cryptanalysis, availibility of $N$ number of bits may not be practical with non-feasible computational complexity. 
While stydying the behaviour of underline binary fields involved in LFSR based combiner generators, certain fixed patterns have been observed whose detailed discussion is given in~\cite{khan2015crt}. Here we discuss the aspects related to spectral convolutions in particular though an example.  

\begin{example}\label{exm-2}
Consider a sequence \textbf{u}$_t$ generated from product of two LFSRs having primitive polynomials of $g_1(x) =  x^2+x+1$ and $g_2(x)  = x^3+x+1$. The period $n_1$ of stream \textbf{a}$_{t}$ corresponding to LFSR-1 is $3$ and $n_2$ of \textbf{b}$_{t}$ corresponding to LFSR-2 is $7$. The period $N$ of \textbf{u}$_t$ is $21$.

\begin{enumerate}
\item In time domain representation, we have following sequences with initial state of '01' and '001' for both LFSRs as:
\begin{enumerate}
\item Sequence \textbf{a}$_t$:\;\;\;$011$\;\;\;\;\;\;\;\;\;\;\;\;\;\;\;\;\;\;\;\;\;\;\;\;\;\;\;\;\;\;\;\;\;\;\;\;\;\;\;\;\;\;\;\;\;\;\;\;\;\;\;\;\;\;\;\;\;\;\;\;\;\;\;\;\;\;\;\;\;\;\;\;\;(of period 3)
\item Sequence \textbf{b}$_t$:\;\;\;$0010111$\;\;\;\;\;\;\;\;\;\;\;\;\;\;\;\;\;\;\;\;\;\;\;\;\;\;\;\;\;\;\;\;\;\;\;\;\;\; \;\;\;\;\;\;\;\;\;\;\;\;\;\;\;\;\;\;\;\;\;\;\;\;\;\;(of period 7)
\item Sequence \textbf{u}$_t$:\;\;\;$001011000001010010011101110111011101$\;\;\;\;\;\;\;\;\;\;\;\;\;\;(of period 21)
\end{enumerate}
\item From equations~(\ref{DFT-A eq}), (\ref{DFT-B eq}) and (\ref{DFT-S eq}) , frequency domain representations of these sequences represented in powers of roots of their associated binary fields as $\alpha \in GF(2^2)$, $\beta \in GF(2^3)$ and $\lambda \in GF(2^6)$, are:
\begin{enumerate}
\item \textbf{$A$} $ = 0,1,1$.
\item \textbf{$B$}  $= 0,0,0,\beta^4, 0,\beta^2,\beta$.
\item To compute \textbf{$S$}, associated minimaum polynomial, determined through Berlekamp-massey algorithm,  is: 
$g(x) = x^6+x^4+x^2+x+1$.\\
\item The DFT of \textbf{u}$_t$ is:
\textbf{$U$} $= 0,0,0,0,0,\lambda^9,0,0,0,0, \lambda^{18},0,0,\lambda^{15},0,0,0,\lambda^{18},0,\lambda^9,\lambda^{15}$.
\end{enumerate}
\end{enumerate} 
\end{example}

To compute DFT of \textbf{u}$_t$, we first determined its associated minimal polynomial through Berlekamp-massey algorithm and then carried out DFT computations in $GF(2^6)$. We have observed that there exists a fixed mapping between elements belonging to base fields of $GF(2^2)$ and $GF(2^3)$ to their product field $GF(2^6)$ which can be exploited to determine spectral componenets of product stream. To illustrate these novel observations, we arrange spectral  components of \textbf{$A$}, \textbf{$B$} and \textbf{$U$} in a Table~\ref{tab:ABC-1} as:

\begin{table}[!ht]
\small
\begin{center}
\caption[Sample Table]{Spectral Components of \textbf{u} = \textbf{a}.\textbf{b}}
\begin{tabular}{|c| c| c| c| c| c| c| c| c| c| c| c| c|}
\hline
%\rowcolor{Gray}
 
 Index & 0 & 1 & 2 & 3 & 4 & 5 & 6 & 7 & 8 & 9 & 10 & 11    \\ \hline

 \textbf{$A$} & 0 & $\alpha^0$ & $\alpha^0$ & 0 & $\alpha^0$ & $\alpha^0$ & 0 & $\alpha^0$ & $\alpha^0$ & 0 & $\alpha^0$ & $\alpha^0$    \\ \hline
 
 \textbf{$B$} & 0 & 0 & 0 & $\beta^4$ & 0 & $\beta^2$ & $\beta$ & 0 & 0 & 0 & $\beta^4$ & 0    \\ \hline \hline

\textbf{$U$} & 0 & 0 & 0 & 0 & 0 & $\gamma^9$ & 0 & 0 & 0 & 0 & $\gamma^{18}$ & 0    \\ \hline \hline
 %&  &  &  &  &  &  &  &  &  &  & &   \\ 
 Index & 12 & 13 & 14 & 15 & 16 & 17 &18 & 19 & 20 &  &  &     \\ \hline 
 
 % &  &  &  &  &  &  &  &  &  &  & &   \\
 \textbf{$A$} & 0 & $\alpha^0$ & $\alpha^0$ & 0 & $\alpha^0$ & $\alpha^0$ & 0 & $\alpha^0$ &$\alpha^0$  &  &  &     \\ \hline
 
 % &  &  &  &  &  &  &  &  &  &  & &   \\
 
 \textbf{$B$} & $\beta^2$ &$\beta$ & 0 & 0 & 0 & $\beta^4$ & 0 & $\beta^2$ & $\beta$ &  &  &     \\ \hline \hline
 
 \textbf{$U$} & 0 & $\lambda^{15}$ & 0 & 0 & 0 & $\lambda^{18}$ & 0 & $\lambda^9$ & $\lambda^{15}$ &  &  &  \\ \hline

\end{tabular}

\label{tab:ABC-1}
\end{center}
\end{table}

It is very clear from Table~\ref{tab:ABC-1} that non-zero spectral component of  \textbf{U} only exists where corresponding DFT points of \textbf{A} and \textbf{B} are non-zero. Moreover, their exists a certain fixed mapping from DFT points of \textbf{A} and \textbf{B} to \textbf{S}.  Theorem~\ref{CRT-Ptheorem} describes the phenomenon of this fixed mapping.

\begin{thm}
\label{CRT-Ptheorem}
Let \textbf{u}$_t\in GF(2^n)$ be a product sequence with period $N  \mid 2^n - 1$ and $0 \leq t \leq N-1$, having two constituent sequences \textbf{a}$_{t\;(mod\;n_1)}$ $ \;\in\; GF(2^p)$ and \textbf{b}$_{t\;(mod\;n_2)}$ $\;\in\; GF(2^q)$ based on primitive polynomials where periods  $n_1\; = \; 2^{(p-1)}$  and  $n_2\; = \;2^{(q-1)}$  are coprime to each other and $p < q$. Let \textbf{$A$} be a DFT spectra of \textbf{a}, \textbf{B} be a DFT spectra of \textbf{b} and \textbf{U} be a DFT spectra of \textbf{u}, any $k^{th}$ spectral component  \textbf{U}$_k$ of \textbf{U},  corresponding to non-zero spectral components of \textbf{A} and \textbf{B}, can be determined directly through CRT as:

      \begin{eqnarray*}
      d_k  &\equiv & d_{k_{1}}\mbox{\; (mod} \mbox{\;}n_1) \\
      d_k  &\equiv & d_{k_{2}}\mbox{\; (mod} \mbox{\;}n_2) \\
               \end{eqnarray*}
      
   where $d_k$, $d_{k_{1}}$ and $d_{k_{2}}$ are degrees of non-zero spectral components of \textbf{$U_k$}, \textbf{$A_{k\;mod(p-1)}$}, and \textbf{$B_{k\;mod(q-1)}$}  represented in terms of associated roots $\lambda \in GF(2^n)$,  $\alpha \in GF(2^{p})$ and $\lambda \in GF(2^{q})$ of    
minimal polynomials of \textbf{u}$_t$, \textbf{a}$_{t\;mod(p-1)}$ and \textbf{b}$_{t\;mod(q-1)}$ respectively.
\end{thm}
 
\begin{proof}

All roots of  minimal polynomials of \textbf{a}$_t$, \textbf{b}$_t$ and \textbf{s}$_t$ lie within their respective fields i.e. $\alpha \in GF(2^{p})$, $\beta \in GF(2^{q})$ and $\lambda \in GF(2^n)$ respectively.
As $n_1$ and $n_2$ are coprime, $N= lcm (n_1, n_2)$.

As \textbf{u}$_t$ = \textbf{a}$_t$.\textbf{b}$_t$. Let $g_1(x)$ generates a sequence \textbf{a}$_t$ of period $2^{p_1}-1$ having zeros $\alpha^{2^i}$ with $0 \leq i \leq p_1 - 1$ where $\alpha$ is an primitive element of order $2^{p_1}-1$. Similarly, let $g_2(x)$ generates a sequence \textbf{b}$_t$ of period $2^{p_2}-1$ having zeros $\beta^{2^j}$ with $0 \leq j \leq p_2 - 1$ where $\beta$ is an primitive element of order $2^{p_2}-1$.

As we know from equation~(\ref{DFT-S eq}),
 \begin{equation}\label{DFT-proof eq}
  U_{k} = \sum_{t=0}^{N-1} s_{t}\lambda^{tk} , \;\;\;  k = 0,1,2,.....,N-1
\end{equation}
where $\lambda$ is the root of polynomial $g(x)$.
 Through Berlekamp-Massey Algorithm, we know that sequence \textbf{u}$_t$ is generated by an LFSR defined over polynomial $g(x)$. In our particular case when $gcd(p, q)=1$, $g(x)$ is irreducible with degree $pq$. In such a case all roots of $g(x)$ can be written in terms of $\lambda = (\alpha \beta)^{2^{l}}$ with $0 \leq l \leq pq - 1$. Thus we can represent $\lambda = \alpha \beta$ in equation~(\ref{DFT-proof eq}) as
 \begin{equation}\label{DFT- eq3}
  U_{k} = \sum_{t=0}^{N-1} s_{t}(\alpha \beta)^{tk} , \;\;\;  k = 0,1,2,.....,N-1
\end{equation}
  As \textbf{u}$_t$ = \textbf{a}$_t$.\textbf{b}$_t$, equation~(\ref{DFT- eq3}) can be written as:
 \begin{equation*}\label{DFT- eq4}
  U_{k} = \sum_{t=0}^{N-1} (a_{t}b_{t} )(\alpha \beta)^{tk} , \;\;\;  k = 0,1,2,.....,N-1
\end{equation*}  
 
 \begin{equation*}\label{DFT- eq5}
 \Rightarrow U_{k} = \sum_{t=0}^{N-1} (a_{t}\alpha^{tk})(b_{t} \beta^{tk}) , \;\;\;  k = 0,1,2,.....,N-1
\end{equation*}
Due to orthogonality and cyclic behaviour of these fields, we will have:
 \begin{equation}\label{DFT- eq6}
 \Rightarrow U_{k} = \sum_{t=0}^{N-1} (a_{t}\alpha^{tk})\sum_{t=0}^{N-1}(b_{t} \beta^{tk}) , \;\;\;  k = 0,1,2,.....,N-1
\end{equation}
 By substituting equation~(\ref{DFT-A eq}) and~(\ref{DFT-B eq}) in (\ref{DFT- eq6}), we get
 \begin{equation}\label{DFT- eq7}
  U_{k} = A^{1}_k  B^{2}_k , \;\;\;  k = 0,1,2,.....,N-1.
\end{equation}
 
Thus from Equation~(\ref{DFT- eq7}), spectral components of \textbf{U} are non-zero at all indices where corresponding spectral components of \textbf{A} and \textbf{B} are non-zero. As all DFT spectral components of  \textbf{U} lie within $GF(2^n)$ and correspond to $\lambda^h$, where $0 \leq h \leq N$. Considering any $k^{th}$ component of \textbf{U} corresponding to non-zero DFT components of \textbf{A} and \textbf{B},  we only need to prove that both non-zero spectral components of \textbf{A} and \textbf{B} have one to one mapping to \textbf{U} through CRT. 
 
We now transform the relationship of \textbf{u}$_t$= \textbf{a}$_t$.\textbf{b}$_t$  into roots of associated polynomials of each sequence in their respective binary fields by representing \textbf{U}$_{k}$ $\in GF(2^m)$ in terms of $\lambda^h$ ($0 \leq h \leq N$), \textbf{A} $ \in GF(2^{p_1})$ in terms of $\alpha^{i}$ ($0\leq i \leq n_1$) and \textbf{B} $\in GF(2^{p_2})$ in terms of $\beta^{j}$ ($0 \leq j \leq n_2$). Thus we have:
\begin{equation}\label{gamma}
  \lambda^{d_k} = \alpha^{d_k}.\beta^{d_k} , \;\;\;  d_k = 0,1,2,.....,N-1.
\end{equation}
As we can write \textbf{u}$_t$= \textbf{a}$_{(t\; mod\; n_1)}$.\textbf{b}$_{(t \;mod\; n_2)}$ and $U_{k} = A_{(k\;mod\; n_1)} B_{(k\;mod \;n_2)}$, Equation~(\ref{gamma}) can be expressed as:

\begin{equation}\label{gamma-mod}
  \lambda^{d_k} = \alpha^{d_{k\;mod\; n_1}}.\;\beta^{d_{k\; mod\; n_2}} , \;\;\;  d_k = 0,1,2,.....,N-1,
\end{equation}
\begin{equation}\label{gamma-mod1}
 \Rightarrow \lambda^{d_k} = \alpha^{d_{k_{1}}}.\;\beta^{d_{k_{2}}} , \;\;\;  d_k = 0,1,2,.....,N-1.
\end{equation}
From equations~(\ref{gamma-mod}) and (\ref{gamma-mod1}), there exists a unique mapping for $k^{th}$ degrees of  $\lambda$, $\alpha$, and $\beta$ which can be computed using CRT as:
\begin{eqnarray*}
      d_k  &\equiv& d_{k_{1}}\;(mod\;n_1) \\
      d_k  &\equiv& d_{k_{2}}\;(mod\;n_2).   
               \end{eqnarray*}
\end{proof}
 Let we relate the results of Example~\ref{exm-2} with Theorem~\ref{CRT-Ptheorem}. Using \textbf{A}$_5$ and \textbf{B}$_5$ from Table~\ref{tab:ABC-1}, we can directly compute the \textbf{U}$_5$ using the Theorem Theorem~\ref{CRT-Ptheorem}. As \textbf{A}$_5 = \alpha^0$ and \textbf{B}$_5 = \beta^2$, we will use CRT as:
 \begin{eqnarray*}
      d_5  &\equiv& 0\;(mod\;3) \\
      d_5  &\equiv& 2\;(mod\;7).   
               \end{eqnarray*}
Thus  \textbf{d}$_5 = 9$ and therefore \textbf{U}$_5 = \lambda^9$. Similarly, all six non-zero spectral points of \textbf{U} at indices $5,\; 10,\; 13,\; 17,\; 19$ and $20$  can be computed directly using CRT relation of Theorem~\ref{CRT-Ptheorem}.

\section{Subspace Decomposition and CRT Based Spectral Convolutions}
In this section, application of our proposed method of computing spectral components for analysis of cryptographic sequences is discussed. Specifically for combiner generators,  relevance of CRT based commputations of DFT points in binary fields to their analysis through subspace decomposition is made. As booelan functions used in combiner generators comprise of different combination of smaller product sequences, our proosed method of computing  DFT helps to analyze the sequences in frequency  domain. Let we consider a simple boolean function as:
\begin{equation*}
 f(x_1,x_2,,x_{3}) = x_1x_2 + x_2x_3 + x_1x_3 \\
\end{equation*}
where $x_1 \in GF(2^p)$, $x_2 \in GF(2^q)$, $x_3 \in GF(2^r)$ and $f(x_1,x_2,x_{3}) \in GF(2^n)$. As the boolean function commprises of three product components, its frequency domain analysis can be based either on its composite form of $f(x_1,x_2,,x_{3})$ or on three sub-components of $x_1x_2$, $x_2x_3$ and $x_1x_3$. Let the resultant stream be $s_t \in GF(2^n)$ having a period $N|2^{n}-1$. To analyze the sequence $s_t$ in frequency domain classically, its associated minimal polynomial $g(x) \in GF(2^n)$ is required to be determined through Berlekamp-Massey algorithm. With primitive element of polynomial $g(x)$, each spectral component \textbf{S}$_k$ is calculated through Equation~(\ref{DFT-S eq}) in $GF(2^n)$. On the contrary, our proposed methodology of CRT based method of computing spectral components can be used to decompose the involved space of $GF(2^n)$ on the basis of its basic component fields of $GF(2^p)$, $GF(2^q)$ and $GF(2^r)$.   The analysis through sub space decomposition in this case reduces the complexity significantly.  To illustrate the idea, let we discuss the details with an example here:-
\begin{example}\label{exm-subspace}
Consider a Boolean function $f(x)$ combining three LFSR sequences \textbf{a} $\in GF(2^2)$ generated with $g_1(x) = x^2+x+1$, \textbf{b} $\in GF(2^3)$ generated with $g_2(x) = x^3+x+1$ and \textbf{c} $\in GF(2^5)$ generated with $g_3(x) = x^5+x^2+1$ to make the combined sequence $s_t \in GF(2^n)$ as:
\begin{equation}
s_t = f(a_t,b_t,c_t) = a_tb_t + b_tc_t + a_tc_t \;\;\;\;\mbox{with}\;\; 0 \leq t \leq N-1 
\end{equation}
where $N = lcm(2^{2-1},2^{3-1}, 2^{5-1})$ which is $651$ in this case here. We generate $651$ bits of $s_t$ and run Berlekamp-Massey algorithm. Linear complexity of \textbf{s}$_t$ is 31 and the corresponding minimum polynomial 
$g(x)$ = $x^{31} + x^{29} + x^{28} + x^{27} + x^{24} + x^{23} + x^{22} + x^{20} + x^{18} + x^{17} + x^{16} + x^{15} + x^{13} + x^{11} + x^{10} + x^9 + x^8 + x^7 + x^5 + x^4 + x^2 + x + 1 $.
 	
 	 As $f(a_t,b_t,c_t) = a_tb_t + b_tc_t + a_tc_t$, we will consider the component streams of \textbf{a}$_t$\textbf{b}$_t$, \textbf{b}$_t$\textbf{c}$_t$ and \textbf{a}$_t$\textbf{c}$_t$ one by one. As spectral points of \textbf{a}$_t$\textbf{b}$_t$ have already been computed in Example~\ref{exm-2}, let we denote it by \textbf{AB}. For \textbf{b}$_t$\textbf{c}$_t$, we will take individual spectra of \textbf{B} and \textbf{C} and will then use our method of computing non-zero spectral  points. Frequency domain representations of these sequences represented in powers of roots of their associated binary fields as $\beta \in GF(2^3)$, $\gamma \in GF(2^5)$ are:
\begin{enumerate}

\item \textbf{B}  $= 0,0,0,\beta^4, 0,\beta^2,\beta$.
\item \textbf{C} $= 0,0,0,0,0,0,0,0,0,0,0,0,0,0,0, \gamma^{29}, 0,0,0,0,0,0,0,\gamma^{30},0,0,0,\gamma^{15},0,\gamma^{23},\gamma^{27} $.
\end{enumerate}
The period $N_{bc} = lcm(7,31) = 217$. The associated minimal polynomial of stream \textbf{b}$_t$\textbf{c}$_t$ is $g_{bc}= x^{15} + x^{12} + x^{10} + x^7 + x^6 + x^2 + 1$.	From famous Blahut's Theorem~\cite{Blahut1983errorcontrol}, number of non-zero spectral components of \textbf{B}\textbf{C} must be $15$. We can directly compute all non-zero DFT points of \textbf{B}\textbf{C} through our CRT based method. Taking $3^{rd}$ non-zero component  of \textbf{B} and $15^{th}$ th non zero DFT point of \textbf{C}, we can determine the corresponding index of non-zero spectral component of \textbf{B}\textbf{C} through CRT as:
\begin{eqnarray*}
      k  &\equiv& 3\; (mod\;7) \\
      k  &\equiv& 15\; (mod\;31).   
               \end{eqnarray*}
               Thus $k = 108$ and value of spectral component  \textbf{BC}$_k$ is again computed by using our method of CRT based DFT points as: 
               \begin{eqnarray*}
      d_{108}  &\equiv& 4\;(mod\;7) \\
      d_{108}  &\equiv& 29\;(mod\;31).   
               \end{eqnarray*}
             We get $ d_{108} = \delta^{60}$. Similarly, all 15x non-zero DFT points of \textbf{BC} are determined using our method in in a Table~\ref{tab:BC} as:

\begin{table}[!ht]
\small
\begin{center}
\caption[Sample Table]{ 15 Spectral Components of \textbf{b}.\textbf{c}}
\begin{tabular}{|c| c| c| c| c| c| c| c| c| }
\hline
%\rowcolor{Gray}
 
 Index & 108 & 178 & 213 & 122 & 185 & 201 & 54 & 89      \\ \hline

 \textbf{B} & $\beta^4$ & $\beta^2$ & $\beta$ & $\beta^4$ & $\beta^2$ & $\beta$ & $\beta^4$ & $\beta^2$     \\ \hline 

\textbf{C} & $\gamma^{29}$ & $\gamma^{30}$ & $\gamma^{15}$ & $\gamma^{23}$ & $\gamma^{27}$ & $\gamma^{29}$ & $\gamma^{30}$ & $\gamma^{15}$     \\ \hline \hline

\textbf{BC} & $\delta^{60}$ & $\delta^{123}$ & $\delta^{46}$ & $\delta^{116}$ & $\delta^{151}$ & $\delta^{184}$ & $\delta^{30}$ & $\delta^{170}$     \\ \hline \hline

Index & 215 & 61 & 139 & 209 & 27 & 153 & 216 &       \\ \hline

 \textbf{B} & $\beta$ & $\beta^4$ & $\beta^2$ & $\beta$ & $\beta^4$ & $\beta^2$ & $\beta$ &      \\ \hline 

\textbf{C} & $\gamma^{23}$ & $\gamma^{27}$ & $\gamma^{29}$ & $\gamma^{30}$ & $\gamma^{15}$ & $\gamma^{29}$ & $\gamma^{30}$ &     \\ \hline \hline

\textbf{BC} & $\delta^{23}$ & $\delta^{58}$ & $\delta^{29}$ & $\delta^{92}$ & $\delta^{15}$ & $\delta^{85}$ & $\delta^{120}$ &      \\ \hline \hline
 %&  &  &  &  &  &  &  &  &  &  & &   \\ 

\end{tabular}

\label{tab:BC}
\end{center}
\end{table}
Now consider \textbf{a}\textbf{c} which has period of $N_{ac} = lcm(3,31)=93$. The associated minimal polynomial of stream \textbf{a}$_t$\textbf{c}$_t$ is $g_{bc}= x^{10} +  x^{5} + x^4 + x^2 + x + 1$.	Number of non-zero spectral components of \textbf{A}\textbf{C} must be $10$. We can directly compute all non-zero DFT points of \textbf{B}\textbf{C} through our CRT based method. Taking $1^{st}$ non-zero component  of \textbf{A} and $15^{th}$ non zero DFT point of \textbf{C}, we can determine the corresponding index of non-zero spectral component of \textbf{B}\textbf{C} through CRT as:
\begin{eqnarray*}
      k  &\equiv& 1\; (mod\;3) \\
      k  &\equiv& 15\; (mod\;31).   
               \end{eqnarray*}
               Thus $k = 46$ and value of spectral component  \textbf{BC}$_k$ is again computed by using our method of CRT based DFT points as: 
               \begin{eqnarray*}
      d_{46}  &\equiv& 0\;(mod\;3) \\
      d_{46}  &\equiv& 29\;(mod\;31).   
               \end{eqnarray*}
            We get  $ d_{46} = \delta^{60}$. Similarly, all 10x non-zero DFT points of \textbf{AC} are determined using our method in Table~\ref{tab:AC} as:

\begin{table}[!ht]
\small
\begin{center}
\caption[Sample Table]{ 10 Spectral Components of \textbf{a}.\textbf{c}}
\begin{tabular}{|c| c| c| c| c| c| c| c| c| c| c| }
\hline
%\rowcolor{Gray}
 
 Index & 46 & 85 & 58 & 91 & 61  & 77 & 23 & 89 & 29 & 92     \\ \hline

 \textbf{A} & $\alpha^0$ & $\alpha^0$ & $\alpha^0$ & $\alpha^0$ & $\alpha^0$  & $\alpha^0$ & $\alpha^0$ & $\alpha^0$ & $\alpha^0$ & $\alpha^0$    \\ \hline 

\textbf{C} & $\gamma^{29}$ & $\gamma^{30}$ & $\gamma^{15}$ & $\gamma^{23}$ & $\gamma^{27}$ & $\gamma^{29}$ & $\gamma^{30}$ & $\gamma^{15}$ & $\gamma^{23}$ & $\gamma^{27}$ \\ \hline \hline

\textbf{AC} & $\eta^{60}$ & $\eta^{30}$ & $\eta^{15}$ & $\eta^{54}$ & $\eta^{27}$ & $\eta^{60}$ & $\eta^{30}$ & $\eta^{15}$ & $\eta^{54}$ & $\eta^{27}$     \\ \hline

\end{tabular}
\label{tab:AC}
\end{center}
\end{table}              
Now we consider spectras of all three sequences togather:
  \begin{enumerate}
\item \textbf{A} $=0, \alpha^0, \alpha^0$.
\item \textbf{B}  $= 0,0,0,\beta^4, 0,\beta^2,\beta$.
\item \textbf{C} $= 0,0,0,0,0,0,0,0,0,0,0,0,0,0,0, \gamma^{29}, 0,0,0,0,0,0,0,\gamma^{30},0,0,0,\gamma^{15},0,\gamma^{23},\gamma^{27} $.
\end{enumerate}           
 With $1^{st}$   non-zero index of \textbf{A}, $3^{rd}$ of \textbf{B} and $15^{th}$ of \textbf{C}, we can determine non-zero index of \textbf{S} as :
 \begin{eqnarray*}
	  k  &\equiv& 1\; (mod\;3) \\   
      k  &\equiv& 3\; (mod\;7) \\
      k  &\equiv& 15\; (mod\;31).   
               \end{eqnarray*}  
Result of CRT is $k = 325$ indicating \textbf{S}$_{325}$ to be non-zero. Now taking $\alpha^0$, $\beta^4$  and $\gamma^{29}$, the value of \textbf{S}$_{325}$ is determined through CRT as:
\begin{eqnarray*}
	  d_k  &\equiv& 0\; (mod\;3) \\   
      d_k  &\equiv& 4\; (mod\;7) \\
      d_k  &\equiv& 29\; (mod\;31).   
               \end{eqnarray*}
Thus \textbf{S}$_{325}$ $=$ $\sigma^{60}$, where $\sigma$ is the root of polynomial $g(x)$. Similarly, corresponding to all non-zero indices of \textbf{A}, \textbf{B} and \textbf{C}, we will determine the spectral components of \textbf{S} through our CRT based methods mentioned at Table~\ref{tab:combiner-sigma}.

\begin{table}[ht]
\small
\begin{center}
\caption[Sample Table]{Non-zero Spectral Points of \textbf{S} with  element $\sigma$ of $g(x)$}
\begin{tabular}{|c| c| c| c| c| c| c| c| c| c| c|  }
\hline
%\rowcolor{Gray}
 Index & 61 & 89 & 122 & 139 & 178 & 185 & 209 & 215 & 244 & 271   \\ \hline
  
Spectral Component & $\sigma^{492}$ & $\sigma^{387}$ & $\sigma^{333}$ &$ \sigma^{246}$ & $\sigma^{123}$ & $\sigma^{585}$ & $\sigma^{309}$ & $\sigma^{240}$ & $\sigma^{15}$ & $\sigma^{30}$   \\ \hline \hline 

Index & 278 & 325 & 356 & 370 & 395 & 418 & 430 & 433 & 461 & 488   \\ \hline
  
Spectral Component & $\sigma^{492}$ & $ \sigma^{60}$ & $\sigma^{246}$ &  $\sigma^{519}$ & $\sigma^{123}$ & $\sigma^{618}$ & $\sigma^{480}$ & $\sigma^{120}$ & $\sigma^{15}$ & $ \sigma^{30} $  \\ \hline \hline  

Index & 523 & 542 & 556 & 587 & 619 & 635 & 643 & 647 & 649 & 650  \\ \hline
  
Spectral Component & $\sigma^{387}$ & $\sigma^{60}$ & $\sigma^{333}$ & $\sigma^{519}$ &  $\sigma^{585}$ & $\sigma^{618}$ & $\sigma^{309}$  & $\sigma^{480}$ & $\sigma^{240}$ & $\sigma^{120}$  \\ \hline
\end{tabular}

\label{tab:combiner-sigma}
\end{center}
\end{table}
\end{example}
Correlating subspace components of \textbf{AB}, \textbf{BC} and \textbf{AC} to \textbf{S} from tables above, certain fixed mapping is observed between the spectral components from composite fields componnents to its constituent sub field components. For instance, decommposition of \textbf{S}$_{325}$ into its subspace components  \textbf{AB}$_{\{10\}}$,  \textbf{BC}$_{\{108\}}$ and  \textbf{AC}$_{\{46\}}$and then further to \textbf{A}$_1$, \textbf{B}$_3$ and \textbf{C}$_{15}$ is depicted in Figure~\ref{fig:mapping}. This fixed mapping is considered very useful for exploitation during analysis of the combinatorial sequences. 

\begin{figure}[H]\label{fig:mapping}
	  \centering
		      \includegraphics[scale=0.50]{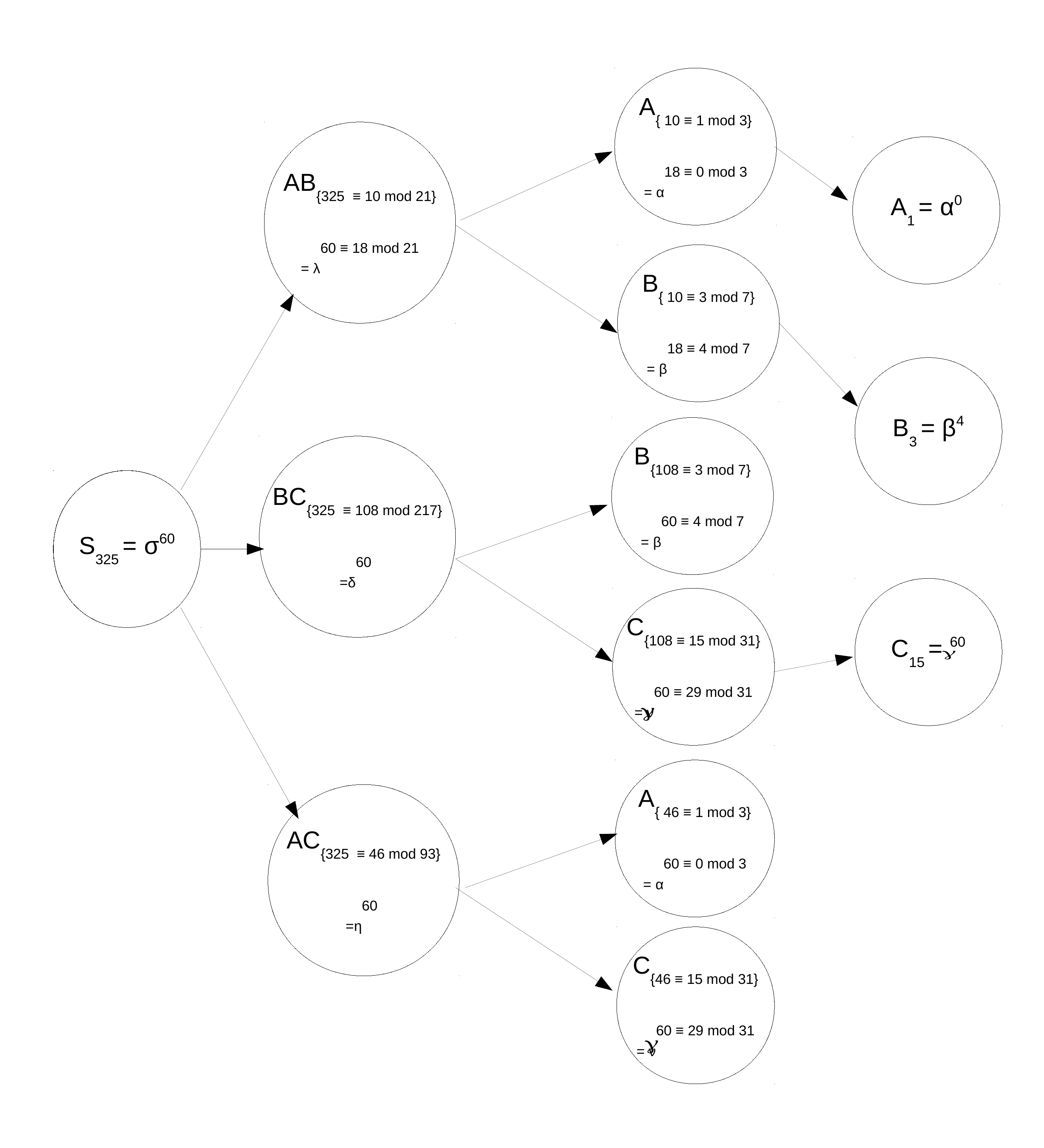}
	  \caption{Fixed Mapping in Subspaces of Spectral Components}
	\end{figure}

The conjugate property~\cite{golomb2005signal} of spectral sequence \textbf{S} can be verified from Figure~\ref{fig:conjugates} below, where trail of only fifteen components is shown. All other Spectral components in Table~\ref{tab:combiner-sigma} follow the same trail in succession. The advantage of this magical behaviour of DFT components in binary fields is drastic reduction in computations required for complete spectra \textbf{S}. In Table~\ref{tab:combiner-sigma} above,  we need to  compute \textbf{S}$_{61}$ only and spectra for all other indices can be computed  by conjugate operation.

\begin{figure}[H]\label{fig:conjugates}
	  \centering
		      \includegraphics[scale=0.80]{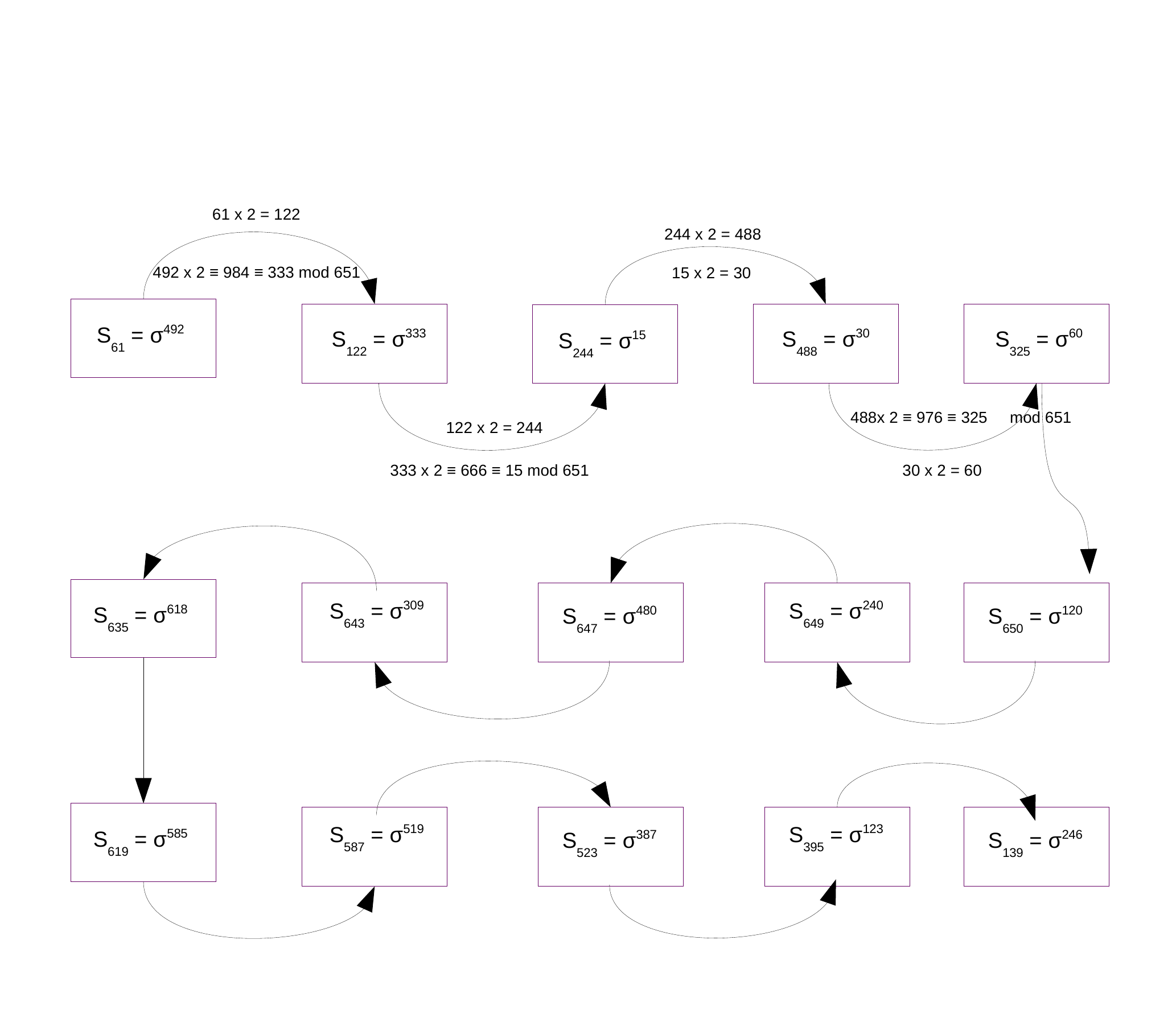}
	  \caption{Conjugate Property of Spectral Sequence \textbf{S} }
	\end{figure}

\section{Complexity Estimations}

Deatiled account of complexity of computing DFT in binary fields can be found in~\cite{gongcloser}, we reuse the results discussed therein to draw commparison of our proposed methodology in case of combinatorial sequences.
In terms of polynomail operations, DFT expression of Equation~(\ref{DFT- eq6}) can be expressed showing the relatioship between DFT and associated minimal polynomial as:

\begin{equation}\label{poly-DFT}
S_k = s(\gamma^{-k}),\;\;\; k = 0,1,\cdots,N-1
\end{equation}

where $s(x) = \sum^{N-1}_{t=0}s_tx^t$. The complexity for computing any $S_k$ using the Equation~(\ref{poly-DFT}) to evaluate s(x) at $\gamma^{-k}$ is determined as follows:-
\begin{enumerate}
\item The complexity for computing $\gamma^{-k}$ is $\mathcal{O}(\;log(k)\;\eta(n)\;)$ Xor operations,  where 
$\eta(n) = [\; n\; log\; n \;log log \;n\;]$ for two polynomials of degree $n$.

\item The complexity to evaluate $s(x)$ at $\gamma^{-k}$ is $\mathcal{O}(\;deg(s(x))\;\eta(n)\;)$ Xor operations.

\item Since the degree of $s(x)$ is on average $\leq N/2 $ and $log(k)\leq degree(s(x))$, the total complexity of computing any $S_k$ is:
\begin{equation}\label{eq:total complexity poly method}
\mathcal{O}[\,(\;log(k)\;\eta(n)\;)\;+\;(\;deg(s(x))\;\eta(n)\;)]\,\;\approx \mathcal{O}(N/2\;\eta(n))\;\mbox{Xor operations}.
\end{equation}
\end{enumerate}

Now,  we determine complexity of CRT based method of computing any DFT point \textbf{S}$_k$ for sequence \textbf{s}$_t\;=\;$\textbf{a}$_t$.\textbf{b}$_t$ with $0 \leq t \leq N-1$ and \textbf{a}$\in GF(2^p)$ , \textbf{b}$\in GF(2^q)$ and \textbf{s}$\in GF(2^n)$ as assumed in Section 2. As from Equation~(\ref{DFT-proof eq}), \textbf{S}$_k$ is computed through CRT relationship of \textbf{A}$_{k\; mod\; n_1}$ and \textbf{B}$_{k\; mod\; n_2}$, we have following relations of computations:
\begin{enumerate}
\item The complexity of computing \textbf{A}$_{k_1 = k\; mod\; n_1}$ is 
\begin{equation*}
\mathcal{O}[\,(\;log(k_1)\;\eta(p)\;)\;+\;(\;deg(a(x))\;\eta(p)\;)]\,\;\approx \mathcal{O}(n_1/2\;\eta(a))\;\mbox{Xor operations}.
\end{equation*}

\item The complexity of computing \textbf{B}$_{k_2 = k\; mod\; n_2}$ is 
\begin{equation*}
\mathcal{O}[\,(\;log(k_2)\;\eta(q)\;)\;+\;(\;deg(b(x))\;\eta(p)\;)]\,\;\approx \mathcal{O}(n_2/2\;\eta(q))\;\mbox{Xor operations}.
\end{equation*}

\item The computational cost for CRT is $\mathcal{O}(len(N)^2)$ where $len =$ number of bits required for representation of $N$.

\item Total complexity of computing \textbf{S}$_k$ through CRT based method is:

\begin{equation}\label{eq:CRT based complexity}
\approx \mathcal{O}[ (n_1/2\;\eta(p)) + (n_2/2\;\eta(q)) + (len(N)^2)]\;\mbox{Xor operations}.
\end{equation}
\item Total number of bits required in this case is $n_1\;+ n_2 < n$.
\end{enumerate}

Our results reveal that complexity of CRT based method of computing any DFT component of a combinatorial sequence through Equation~(\ref{eq:CRT based complexity}) is far less than complexity of Equation~(\ref{eq:total complexity poly method}). Let we briefly demonstrate the results through stream \textbf{bc} from Example~\ref{exm-subspace}. Taking first component of Table~\ref{tab:BC} which is \textbf{BC}$_{\{108\}} = \delta^{60} \in GF(2^{15})$, corresponding CRT based constituent spectral points are \textbf{B}$_{3} = \beta^{4} \in GF(2^{3})$ and \textbf{C}$_{15} = \gamma^{29} \in GF(2^{5})$.  Normally the spectral point \textbf{BC}$_{\{108\}}$ can be computed in $GF(2^{15})$ for which complete $217$ bits are required by using Equation~(\ref{DFT-S eq}). On the other hand, our proposed CRT based method of direct calculations of spectral points use the constituent DFT components in $GF(2^3)$ and $GF(2^5)$.  Comparison of complexities of these two methods for a case scenerio of \textbf{BC}$_{\{108\}}$ is given in Table~\ref{tab:compl-complexity} which clearly shows that the CRT based DFT method  is efficient than classical DFT computations in binary fields for combinatorial sequences.

\begin{table}[H]
\small
\begin{center}
\caption[Sample Table]{Comparison of Computational Complexities for DFT component \textbf{BC}$_{\{108\}}$}
\begin{tabular}{|c| c| c| }
\hline
%\rowcolor{Gray}
  			& DFT based on Equation~(\ref{DFT-S eq}) &	CRT based DFT \\ \hline
   Number of Bits Required & 217 & 7 and 31 \\ \hline
  
Field & $GF(2^{15})$  & $GF(2^3)$ and $GF(2^5)$  \\ \hline  
&      &    For \textbf{B}$_{3} \approx$  $\mathcal{O}(7)$\\ 
 &      &    For \textbf{C}$_{15} \approx$  $\mathcal{O}(218)$\\
Total Complexity &   $\approx \mathcal{O}(12,500)$   &    For CRT step $\approx \mathcal{O}(64)$\\ 
 &  & $\mathcal{O}(7 + 218 + 64)$  \\
		& 			& $\approx \mathcal{O}(289)$ \\ \hline
\end{tabular}
\label{tab:compl-complexity}
\end{center}
\end{table}

\section{Conclusion}
In this paper, new method of computing convolution in frequency domain is presented for combinatorial sequences in binary fields. A simplest case of product of LFSR sequences being a fundamental block of any non-linear Boolean function is considered to demonstrate our results on convolution through DFT in binary fields.  CRT based novel approach to determine DFT points for combinatorial sequences has been illustrated  with associated mathematical rationale. With regard to analysis of combiner generators through subspace decomposition, applicability of our proposed methodology of computing spectral points  is made. We presented certain fixed mapping between the spectral components from composite fields to its decomposed subfield components and highlighted inherent structures in cyclic nature of finite fields which can be exploited during analysis of combiners. The comparison of our proposed CRT based methodology to known theory of DFT computations is discussed and it is proven that proposed CRT based method to compute convolution in binary fields is efficient than exiting methods of DFT computations. 

\bibliographystyle{unsrt}
\bibliography{ff}
%\appendix
%\appendixpage
%\addappheadtotoc

\end{document}